\documentclass[journal,letterpaper,final,onecolumn]{IEEEtran}
\usepackage[utf8]{inputenc}
\usepackage{amsmath}
\usepackage{amsthm}
\newtheorem{theorem}{Theorem}
\newtheorem{proposition}{Proposition}

\newtheorem{lemma}{Lemma}

\newtheorem{example}{Example}
\newtheorem{remark}{Remark}
\title{Polarization for arbitrary discrete memoryless channels}
\author{Eren \c Sa\c so\u glu, Emre Telatar, Erdal Ar{\i}kan}
\newcommand{\defn}{\mathrel{\stackrel{\Delta}{=}}}
\def\bs{\mathbf{s}}
\def\cX{\mathcal{X}}
\def\cY{\mathcal{Y}}
\def\cU{\mathcal{U}}

\def\cF{\mathcal{F}}
\def\cP{\mathcal{P}}

\def\Zm{Z_{\text{max}}}

\def\minus{\mathord{-}}
\def\plus{\mathord{+}}
\begin{document}
\maketitle

\begin{abstract}
Channel polarization, originally proposed for binary-input channels, is generalized to arbitrary discrete memoryless channels. Specifically, it is shown that when the input alphabet size is a prime number, a similar construction to that for the binary case leads to polarization. This method can be extended to channels of composite input alphabet sizes by decomposing such channels into a set of channels with prime input alphabet sizes. It is also shown that all discrete memoryless channels can be polarized by randomized constructions. The introduction of randomness does not change the order of complexity of polar code construction, encoding, and decoding. A previous result on the error probability behavior of polar codes is also extended to the case of arbitrary discrete memoryless channels. The generalization of polarization to channels with arbitrary finite input alphabet sizes leads to polar-coding methods for approaching the true (as opposed to symmetric) channel capacity of arbitrary channels with discrete or continuous input alphabets.
\end{abstract}
\begin{keywords} Capacity-achieving codes, channel polarization, polar codes.
\end{keywords}

\section{Polarization}
\label{sec:review}

Channel polarization was introduced in~\cite{Arikan2009} for binary
input discrete memoryless channels as a coding technique to construct
codes --- called polar codes --- for data transmission.  Polar codes are capable of achieving the `symmetric
capacity' of any binary input channel, using low-complexity encoding and decoding algorithms. In terms of the block-length $N$, polar codes can be encoded and decoded in complexity $O(N\log N)$ and achieve a block error probability that decays roughly like $2^{-\sqrt{N}}$. The latter result was shown in~\cite{ArikanTelatar2009}.

The aim of this note is to extend these results
of~\cite{Arikan2009,ArikanTelatar2009} to DMCs with $q$-ary inputs
for any finite integer $q\geq2$.  To that end, we recall the polarization
construction and outline how the results above were shown.

Given a binary input channel $W:\cX\to\cY$ with $\cX=\{0,1\}$ define its
\emph{symmetric capacity} as
\begin{equation}
I(W)=\sum_{x\in\cX}\sum_{y\in\cY}
	\tfrac12 W(y|x)
	\log_2\frac{W(y|x)}{\sum_{x'\in\cX}\frac12 W(y|x')}.
\end{equation}
$I(W)$ is nothing but the mutual information developed between the input
and the output of the channel when the input is uniformly distributed.
In~\cite{Arikan2009}, two independent copies of $W$ are
first combined and then split so as to obtain two \emph{unequal} binary
input channels $W^-$ and $W^+$.  The channel $W^2:\cX^2\to\cY^2$
describes two uses of the channel $W$,
$$
W^2(y_1,y_2|x_1,x_2)=W(y_1|x_1)W(y_2|x_2).
$$
The input $(x_1,x_2)$ to the channel $W^2$ are put in one-to-one
correspondence with $(u_1,u_2)\in\cX^2$ via $x_1=(u_1+u_2) \bmod 2$,
$x_2=u_2$, thus obtaining the \emph{combined} channel $W_2:\cX^2\to\cY^2$ described by
$$
W_2(y_1,y_1|u_1,u_2)=W^2(y_1,y_2|u_1+u_2,u_2)=W(y_1|u_1+u_2)W(y_2|u_2).
$$
The \emph{split} is inspired by the chain rule of mutual information: Let $U_1,U_2,X_1,X_2,Y_1,Y_2$ be random variables corresponding to their lowercase versions above. If
$U_1,U_2$ are independent and uniformly distributed, then so
are $X_1,X_2$ and consequently, on the one hand,
\begin{align*}
I(U_1,U_2;Y_1,Y_2)&=I(X_1,X_2;Y_1,Y_2)=I(X_1;Y_1)+I(X_2;Y_2)=2I(W),
\intertext{and on the other}
I(U_1,U_2;Y_1,Y_2)&=I(U_1;Y_1,Y_2)+I(U_2;Y_1,Y_2,U_1).
\end{align*}
The split channels $W^-$ and $W^+$ describe those that occur on the right hand
side of the equation above:
\begin{align*}
W^-(y_1,y_2|u_1)&=\sum_{u_2\in\cX}\tfrac12 W(y_1|u_1+u_2)W(y_2|u_2),\\
W^+(y_1,y_2,u_1|u_2)&=\tfrac12 W(y_1|u_1+u_2)W(y_2|u_2),
\end{align*}
so that $I(U_1;Y_1,Y_2)=I(W^-)$ and $I(U_2;Y_1,Y_2,U_1)=I(W^+)$.

The polarization construction given in~\cite{Arikan2009} is obtained by
a repeated application of $W\mapsto (W^-,W^+)$.  Since both $W^-$ and $W^+$
are binary input channels, one can obtain $W^{--}:=(W^-)^-$,
$W^{-+}:=(W^-)^+$, $W^{+-}:=(W^+)^-$, and $W^{++}:=(W^{+})^+$.  After
$n$ levels of application, one obtains $2^n$ channels
$W^{-\,\dotsm\,-},\dots,W^{+\,\dotsm\,+}$. 
The main observation
in~\cite{Arikan2009} is that these channels polarize in the following
sense:
\begin{proposition}[\cite{Arikan2009}]
\label{prop:bipolar}
For any $\delta>0$,
\begin{equation}
\label{eq:polar}
\lim_{n\to\infty}
	\frac{\#\bigl\{\bs\in\{+,-\}^n\colon
		I(W^\bs)\in(\delta,1-\delta)\bigr\}
	}{2^n}=0.
\end{equation}
\end{proposition}
In other words, except for a vanishing fraction, all the channels
obtained at level $n$ are either almost perfect, $I(W^\bs)\geq1-\delta$,
or almost pure noise, $I(W^\bs)\leq\delta$.

As the equality $I(W^-)+I(W^+)=2I(W)$ leads by induction to
$\sum_{\bs\in\{+,-\}^n} I(W^\bs)=2^n I(W)$, one then concludes that
the fraction of almost perfect channels approaches the symmetric capacity.
This last observation is 
the basis of 
what lets~\cite{Arikan2009} conclude that polar
codes achieve the symmetric capacity.

We give here a new proof of this proposition because it will
readily generalize to the $q$-ary input case we will discuss later.
Before we embark on this proof, we introduce the Bhattacharyya
parameter for a binary input channel $W:\cX\to\cY$, defined by
\begin{equation}
\label{eq:z-binary}
Z(W)=\sum_y\sqrt{W(y|0)W(y|1)}.
\end{equation}
The relationship between $Z(W)$, $Z(W^-)$, $Z(W^+)$ and $I(W)$ is already
discussed in \cite{Arikan2009}, where the following is shown:
\begin{lemma}[\cite{Arikan2009}]
\label{lem:iz}
\mbox{}\par
\begin{enumerate}
\item[(i)]
	$Z(W^+)=Z(W)^2$,
\item[(ii)]
	$Z(W^-)\leq 2Z(W)-Z(W)^2$, 
\item[(iii)]
	$I(W)+Z(W)\geq 1$,
\item[(iv)]
	$I(W)^2+Z(W)^2\leq 1$.
\end{enumerate}
\end{lemma}

Proposition~\ref{prop:bipolar} was proved in \cite{Arikan2009} for the binary case ($q=2$) using Lemma~\ref{lem:iz}.
Unfortunately, Lemma~\ref{lem:iz} does not generalize to the non-binary case ($q\ge 3$). The following alternate proof of Proposition~\ref{prop:bipolar} uses less stringent conditions that can be fulfilled for all $q\ge 2$.
\begin{lemma}
\label{lem:conv}
Suppose $B_i$, $i=1,2,\dots$ are i.i.d., $\{+,-\}$-valued random variables
with
$$
P(B_1= \minus)=P(B_1 = \plus)=\tfrac12
$$
defined on a probability space $(\Omega,\cF,P)$.  Set $\cF_0=\{\phi,\Omega\}$
as the trivial $\sigma$-algebra and set $\cF_n$, $n\geq1$ to be the
$\sigma$-field generated by $(B_1,\dots,B_n)$.

Suppose further that two stochastic processes $\{I_n:n\geq0\}$ and
$\{T_n:n\geq0\}$ are defined on this probability space with the following
properties:
\begin{itemize}[\labelwidth=2.5em]
\item[(i.1)]
$I_n$ takes values in the interval $[0,1]$ and is measurable with respect to
$\cF_n$.  That is, $I_0$ is a constant, and $I_n$ is a function of
$B_1,\dots,B_n$.
\item[(i.2)]
$\{(I_n,\cF_n):n\geq0\}$ is a martingale.
\item[(t.1)]
$T_n$ takes values in the interval $[0,1]$ and is measurable with respect to
$\cF_n$.
\item[(t.2)]
$T_{n+1}=T_n^2$ when $B_{n+1}=\plus$.
\item[(i\&t.1)]
For any $\epsilon>0$ there exists $\delta>0$ such that
$I_n\in(\epsilon,1-\epsilon)$ implies $T_n\in(\delta,1-\delta)$.
\end{itemize}

Then, $I_\infty:=\lim_{n\to\infty}I_n$ exists with probability 1, $I_\infty$
takes values in $\{0,1\}$, and $P(I_\infty=1)=I_0$.
\end{lemma}
\begin{proof}
The almost sure convergence of $I_n$ to a limit follows from $\{I_n\}$ being
a bounded martingale.  Once it is known that $I_\infty$ is $\{0,1\}$-valued
it will then follow from the martingale property that
$P(I_\infty=1)=E[I_\infty]=I_0$.  It thus remains to prove that $I_\infty$
is $\{0,1\}$-valued.  This, in turn, is equivalent to showing that
for any $\eta>0$,
$$
P\bigl(I_\infty\in(\eta,1-\eta)\bigr)=0.
$$
Since for any $0<\epsilon<\eta$, the event
$\bigl\{I_\infty\in(\eta,1-\eta)\bigr\}$ is included in the event
$$
J_\epsilon:=\bigl\{\omega\colon
	\text{there exists $m$ such that for all $n\geq m$,
		$I_n\in(\epsilon,1-\epsilon)$}\bigr\},
$$
and since by property~(i\&t.1)
there exists $\delta>0$ such that $J_\epsilon\subset K_\delta$ where
$$
K_\delta:=\bigl\{\omega\colon
	\text{there exists $m$ such that for all $n\geq m$,
		$T_n\in(\delta,1-\delta)$}\bigr\},
$$
it suffices to prove that $P(K_\delta)=0$ for any $\delta>0$.
This is trivially true for $\delta\geq1/2$. Therefore, it suffices to show the claim for $0<\delta<1/2$. Given such a $\delta$, find a positive integer $k$
for which $(1-\delta)^{2^k}<\delta$.  This choice of $k$ guarantees that
if a number $x\in[0,1-\delta]$ is squared $k$ times in a row, the result
lies in $[0,\delta)$.

For $n\geq1$ define $E_n$ as the event that $B_n=B_{n+1}=\dots=B_{n+k-1}=\plus$,
i.e., $E_n$ is the event that there are $k$ consecutive $\plus$'s in the sequence
$\{B_i:i\geq1\}$ starting at index $n$.  Note that $P(E_n)=2^{-k}>0$, and
that $\{E_{mk}:m\geq1\}$ is a collection of independent events.  The
Borel--Cantelli lemma thus lets us conclude that the event
\begin{align*}
E&=\{\text{$E_n$ occurs infinitely often}\}\\
&=\{\omega\colon
	\text{for every $m$ there exists $n\geq m$ such that
	$\omega\in E_n$}\}
\end{align*}
has probability 1, and thus $P(K_\delta)=P(K_\delta\cap E)$.  We will now
show that $K_\delta\cap E$ is empty, from which it will follow that
$P(K_\delta)=0$.  To that end, suppose $\omega\in K_\delta\cap E$.  Since
$\omega\in K_\delta$, there exists $m$ such that $T_n(\omega)\in
(\delta,1-\delta)$ whenever $n\geq m$.  But since $\omega\in E$ there
exists $n_0\geq m$ such that $B_{n_0+1}=\dots=B_{n_0+k-1}=\plus$, and thus
$T_{n_0+k}(\omega)=T_{n_0}(\omega)^{2^k}\leq(1-\delta)^{2^k}<\delta$
which contradicts with $T_{n_0+k}(\omega)\in(\delta,1-\delta)$.
\end{proof}

\begin{remark}
The proof of Lemma~\ref{lem:conv} uses property~(t.2) only
in the way that repeated squarings of a number in
$(\delta,1-\delta)$ will eventually fall outside $(\delta,1-\delta)$.
Thus, condition (t.2) may be replaced by any other that has this
property.  E.g., conditioned on $\cF_n$, \emph{at least one} of the
two values of $T_{n+1}$ satisfies
$$
T_{n+1}\leq f(T_n)
$$
for a nondecreasing $f$ having the property that for
any $\delta>0$, there exists $k$ such that $f^{(k)}(1-\delta)\leq\delta$.
Here $f^{(k)}$ denotes $k$-fold composition of $f$.
\end{remark}

\begin{proof}[Proof of Proposition~\ref{prop:bipolar}]
Let $B_1,B_2,\dots$ be i.i.d., $\{+,-\}$-valued random variables taking
the two values with equal probability, as in Lemma~\ref{lem:conv}.  Define
$$
I_n:=I_n(B_1,\dots,B_n)=I(W^{B_1,\dots,B_n})
$$
and
$$
T_n:=T_n(B_1,\dots,B_n)=Z(W^{B_1,\dots,B_n}).
$$
These processes satisfy the conditions of Lemma~\ref{lem:conv}:
(i.1) is trivially true with $I_0=I(W)$; the martingale property (i.2)
follows from $I(W^-)+I(W^+)=2I(W)$; (t.1) is again trivially true;
(t.2) follows from Lemma~\ref{lem:iz}(i); (i\&t.1) follows from
Lemma~\ref{lem:iz}(iii) and~(iv).

Thus, the process $I_n$ converges with probability 1 to a $\{0,1\}$-valued
random variable.  This implies that
$$
\lim_{n\to\infty}P(I_n\in(\delta,1-\delta))=0.
$$
Note that the distribution of $(B_1,\dots,B_n)$ is the uniform
distribution on $\{+,-\}^n$.  Thus,
$$
P(I_n\in(\delta,1-\delta))=\frac{\#\bigl\{\bs\in\{\plus,\minus\}^n\colon
	I(W^\bs)\in(\delta,1-\delta)\bigr\}}{2^n},
$$
and Proposition~\ref{prop:bipolar} follows.
\end{proof}

The following lemma was proved in~\cite{ArikanTelatar2009}.
\begin{lemma}[\cite{ArikanTelatar2009}]
\label{lem:rate}
Suppose that the processes $\{B_n\}$, $\{I_n\}$ and $\{T_n\}$, in addition
to the conditions (i.1), (i.2), (t.1), (t.2) and (i\&t.1) in
Lemma~\ref{lem:conv}, also satisfy
\begin{itemize}[\labelwidth=2.5em]
\item[(t.3)] 
For some constant $\kappa$, $T_{n+1}\leq \kappa T_n$ when $B_{n+1}=\minus$.
\item[(i\&t.2)]
For any $\epsilon>0$ there exists $\delta>0$ such that
$I_n>1-\delta$ implies $T_n<\epsilon$.
\end{itemize}
Then, for any $0<\beta<1/2$
\begin{align}
\label{eq:rate}
\lim_{n\to\infty} P(T_n\leq 2^{-2^{\beta n}}) = I_0.
\end{align}
\end{lemma}
Note that in the proof of Proposition~\ref{prop:bipolar} the random variable $T_n$ denotes the Bhattacharyya parameter of a randomly chosen channel after $n$ steps of polarization.  Therefore, Lemma~\ref{lem:rate} states that after $n$ steps of polarization, almost all `good' channels will have Bhattacharyya parameters that are smaller than $2^{-2^{n\beta}}$ for any $\beta<1/2$, provided that $n$ is sufficiently large. Since the Bhattacharyya parameter is an upper bound to the error probability of uncoded transmission, this implies that
, at any fixed coding rate below $I_0=I(W)$, 
the block error probability $P_e$ of binary polar codes under successive cancellation decoding will satisfy
\begin{align}
\label{eq:Perror}
P_e\leq 2^{-N^\beta} \text{ for all } \beta<1/2,
\end{align}
when the block-length $N=2^n$ is sufficiently large.

\section{Polarization for $q$-ary input channels}
\label{sec:q-ary}

In this section we will show how the transformation $(u_1,u_2)\mapsto
(x_1,x_2)$ (and consequently $W\mapsto(W^-,W^+)$) and the definition
of $Z(W)$ can be modified so that the hypotheses of Lemmas~\ref{lem:conv}
and~\ref{lem:rate} are satisfied when the channel input alphabet is not
binary.  This will establish that the new transformation satisfies
equation~\eqref{eq:polar}, leading to the conclusion that $q$-ary
polar codes achieve symmetric capacity. That the error
probability behaves roughly like $2^{-\sqrt{N}}$ will also follow.

To that end, let
$q$ denote the cardinality of the channel input alphabet $\cX$ and
define
$$
I(W) \defn \sum_{x\in \cX} \sum_{y\in \cY}
	\frac{1}{q} W(y|x) \log
		\frac{W(y|x)}{\sum_{x'\in \cX}\frac{1}{q} W(y|x')}
$$
as the symmetric capacity of a channel $W$.
We will take the base of the logarithm in this mutual information
equal to the input alphabet size $q$, so that $0\leq I(W)\leq 1$.

For any pair of input letters $x,x'\in \cX$, we define the Bhattacharyya distance between them as
\begin{align}
Z(W_{\{x,x'\}}) & = \sum_{y\in \cY} \sqrt{W(y|x)W(y|x')}. \label{defn:zparameter1}
\end{align}
Here, the notation $W_{\{x,x'\}}$ should be interpreted as denoting the channel obtained by restricting the input alphabet of $W$ to the subset $\{x,x'\}\subset \cX$.
We also define the average Bhattacharyya distance of $W$ as
\begin{align}
Z(W) & = \sum_{x,x'\in \cX,x\neq x'} \frac{1}{q(q-1)}Z(W_{\{x,x'\}}). \label{defn:zparameter}
\end{align}
The average Bhattacharyya distance upper bounds the error probability of uncoded transmission:

\begin{proposition}\label{Prop:MLerrorBound}
Given a $q$-ary input channel $W$, let $P_e$ denote the error probability of the maximum-likelihood decoder for a single channel use. Then,
$$
P_e \leq (q-1) Z(W).
$$
\end{proposition}
\begin{proof}
Let $P_{e,x}$ denote the error probability of maximum-likelihood decoding when $x\in \cX$ is sent. We have,
\begin{multline*}
P_{e,x} \leq P \bigl(y: W(y\mid x') \geq W(y\mid x) \text{ for some } x'\neq x \mid x \text{ is sent}\bigr) \\
= \kern-1.5em\sum_{\substack{y\colon \exists x'\neq x \\ W(y\mid x') \geq W(y\mid x)}}\kern-1.6em
W(y\mid x)
 \leq \sum_y\kern-1em\sum_{\substack{x'\colon x'\neq x\\ W (y\mid x') \geq W(y\mid x)}}
\kern-1.6em
W(y\mid x)
\leq \sum_y\sum_{x'\colon x'\neq x}\sqrt{W(y\mid x)W(y\mid x')}.
\end{multline*}
Therefore the average error probability is bounded as
\begin{align*}
P_e = \frac 1q \sum_{x\in\cX} P_{e,x} \leq \frac 1q \sum_{x\in\cX} \sum_{x'\neq x} \sum_y \sqrt{W(y\mid x)W(y\mid x')} = (q-1) Z(W).
\end{align*}

\end{proof}

\begin{proposition}\label{rateReliability}
We have the following relationships between $I(W)$ and $Z(W)$.
\begin{align}\label{Eq:rateReliability}
I(W) &\ge \log \frac{q}{1+(q-1) Z(W)}\\[.5ex]
\label{Eq:rateReliability4}
I(W) &\le \log (q/2) +(\log2)\sqrt{1-Z(W)^2}\\[.5ex] 
\label{Eq:rateReliability2}
I(W) & \le 2(q-1)(\log e)\sqrt{1-Z(W)^2}. 
\end{align}
\end{proposition}

Proof is given in the Appendix.

\subsection{Special case: Prime input alphabet sizes}
\label{subsec:prime}
We will see that when the input alphabet size $q$ is a prime number, polarization can be achieved by similar constructions to the one for the binary case. For this purpose, we will equip the input alphabet $\cX$ with an operation `$+$' so that
$(\cX,+)$ forms a group. (This is possible whether or not 
$q$ is prime.) We will let $0$ denote the identity element of $(\cX,+)$. In particular, we may assume that $\cX=\{0,\dotsc,q-1\}$ and that `$+$' denotes modulo-$q$ addition.  Note that when $q$ is prime, this is the only group of order $q$.

As in the binary case, we combine two independent copies of $W$, by choosing the input to each copy as
\begin{align}
\label{eq:q-combine}
\begin{split}
x_1&=u_1+u_2, \\
x_2&=u_2.
\end{split}
\end{align}
We define the channels $W^-$ and $W^+$ through
\begin{align}
\label{eq:q-split}
\begin{split}
W^-(y_1,y_2\mid u_1)&=\sum_{u_2\in\cX} \frac1q W_2(y_1,y_2\mid u_1,u_2) \\
W^+(y_1,y_2,u_1\mid u_2)&=\frac1q W_2(y_1,y_2\mid u_1,u_2),
\end{split}
\end{align}
where again $W_2(y_1,y_2\mid u_1,u_2)=W(y_1\mid u_1+u_2) W(y_2\mid u_2)$.

The main result of this section is the following:
\begin{theorem}\label{thm:prime}
The transformation described in \eqref{eq:q-combine} and \eqref{eq:q-split} polarizes all $q$-ary input channels in the sense of Proposition~\ref{prop:bipolar}, provided that $q$ is a prime number. The rate of polarization under this transformation is the same as in the binary case, in the sense that the block error probabilities of polar codes based on this transformation satisfy \eqref{eq:Perror}.
\end{theorem}

To prove Theorem~\ref{thm:prime} we first rewrite $Z(W)$ as
$$
Z(W) = \frac{1}{q-1} \sum_{d\neq0} Z_d(W),
$$
where we define
$$
Z_d(W) \defn \frac1q \sum_{x\in\cX} Z(W_{\{x,x+d\}}), \qquad d\neq 0.
$$
We also define
$$
\Zm(W)\defn \max_{d\neq0} Z_d(W).
$$
We will use the following lemma in the proof.
\begin{lemma}\label{lem:Zmax}
Given a channel $W$ whose input alphabet size $q$ is prime, if $\Zm(W)\geq 1-\delta$, then $Z(W)\geq 1-q(q-1)^2\delta$ for all $\delta>0$. 
\end{lemma}

\begin{proof}
Let $d$ be such that $\Zm(W)=Z_d(W)$, and note that $Z_d(W)\geq 1-\delta$ implies
$$
1-Z(W_{\{x,x+d\}}) \leq q\delta \quad \text{for all } x\in\cX.
$$
For a given $x\in\cX$ define
\begin{align*}
a_y & = \sqrt{W(y\mid x)}-\sqrt{W(y\mid x+d)}, \\
b_y & = \sqrt{W(y\mid x+d)}-\sqrt{W(y\mid x+d+d)}. 
\end{align*}
for all $y\in\cY$. The triangle inequality states that
$$
\left(\sum_y (a_y+b_y)^2\right)^{\kern-3pt 1/2} \leq \left(\sum_y a_y^2\right)^{\kern-3pt 1/2} + \left(\sum_y b_y^2\right)^{\kern-3pt 1/2},
$$
or equivalently, that
\begin{align}
\begin{split}
\label{eq:Z-triangle}
\sqrt{1-Z(W_{\{x,x+d+d\}})} & \leq 
	\sqrt{1-Z(W_{\{x,x+d\}})} + \sqrt{1-Z(W_{\{x+d,x+d+d\}})}\\
& \leq 2\sqrt{q\delta}.
\end{split}
\end{align}
On the other hand, since $q$ is prime, the input alphabet can be written as
$$
\cX=\big\{x,x+d,x+d+d,\dotsc,x+\underbrace{d+\dotsm+d}_{q-1\text{ times}}\big\}
$$
for any $d\neq0$ and $x\in\cX$. Hence, applying inequality \eqref{eq:Z-triangle} repeatedly yields
$$
\sqrt{1-Z(W_{\{x,x'\}})} \leq (q-1) \sqrt{q\delta}
$$
for all $x,x'\in\cX$, which implies
$$
Z(W) = \frac{1}{q(q-1)} \sum_{x,x':x\neq x'} Z(W_{\{x,x'\}}) \geq 1-q(q-1)^2\delta.
$$
\end{proof}

\begin{proof}[Proof of Theorem~\ref{thm:prime}]
The proof is similar to the one for the binary case: Let $B_1,B_2,\dotsc$ be i.i.d.\ $\{+,-\}$-valued random variables taking the two values with equal probability. Define the random processes
$$
I_n := I_n(B_1,\dotsc,B_n) = I(W^{B_1,\dotsc,B_n})
$$
and
$$
T_n := T_n(B_1,\dotsc,B_n) = \Zm(W^{B_1,\dotsc,B_n}),
$$
with $I_0=I(W)$ and $T_0=\Zm(W)$. It suffices to show that $\{I_n\}$ and $\{T_n\}$ satisfy the conditions of Lemmas~\ref{lem:conv} and \ref{lem:rate}: Conditions (i.1), (i.2), and (t.1) hold trivially. Also, by \eqref{Eq:rateReliability} and \eqref{Eq:rateReliability2} in Proposition~\ref{rateReliability}, for any $\epsilon>0$ there exists $\delta>0$ such that 
$$
I(W)\in (\epsilon,1-\epsilon) \text{ implies } Z(W)\in (\delta,1-\delta).
$$
Furthermore, it follows from Lemma~\ref{lem:Zmax} that for any $\delta>0$ 
$$
Z(W) \in (\delta, 1-\delta) \text{ implies } \Zm(W)\in(\delta, 1-\delta/[q(q-1)^2]),
$$
from which (i\&t.1) follows. To show (t.2), we write
\begin{align*}
Z_d(W^+) & = \frac1q \sum_x Z(W^+_{\{x,x+d\}}) \\
& = \frac1q \sum_x \frac1q \sum_{y_1,y_2,u} \sqrt{W(y_1\mid x+u) 
	W(y_1\mid x+d+u)} \sqrt{W(y_2\mid x) W(y_2\mid x+d)} \\
& = \frac1q \sum_x Z(W_{\{x,x+d\}}) \frac1q \sum_u Z(W_{\{x+u,x+u+d\}}) \\
& = Z_d(W)^2,
\end{align*}
which implies $\Zm(W^+)=\Zm(W)^2$, or equivalently $T_{n+1}=T_n^2$ when $B_{n+1}=+$. Similarly, one can bound $Z_d(W^-)$ as
\begin{align*}
Z_d(W^-) & = \frac1q \sum_x Z(W^-_{\{x,x+d\}}) \\
& = \frac1q \sum_x \sum_{y_1,y_2} \frac1q
	\sqrt{\sum_u W(y_1\mid x+u)W(y_2\mid u)
	\sum_v W(y_1\mid x+d+v)W(y_2\mid v)} \\
& \leq \frac1q \sum_x \sum_{y_1,y_2}\sum_{u,v} \frac1q
	\sqrt{W(y_1\mid x+u)W(y_2\mid u)
	W(y_1\mid x+d+v)W(y_2\mid v)} \\
& = \frac1q \sum_{u}\frac1q \sum_x \sum_{y_1}
	\sqrt{W(y_1\mid x+u)
	W(y_1\mid x+d+u)} \\
& \hspace{5em} + \frac1q \sum_{u,v:u\neq v} 
	\sum_{y_2} \sqrt{W(y_2\mid u) W(y_2\mid v)}
	\frac1q\sum_x \sum_{y_1}\sqrt{W(y_1\mid x+u) W(y_1\mid x+d+v)} \\
& = Z_d(W) + \sum_{\Delta\neq 0} \frac1q \sum_{u}
	\sum_{y_2} \sqrt{W(y_2\mid u) W(y_2\mid u+\Delta)}
	\frac1q\sum_x \sum_{y_1}\sqrt{W(y_1\mid x+u) W(y_1\mid x+d+u+\Delta)} \\
& = 2Z_d(W)+ \sum_{\substack{\Delta\neq0\\ d+\Delta\neq0}} Z_\Delta(W)Z_{d+\Delta}(W) \\
& \leq 2Z_d(W)+(q-2)\Zm(W)^2.
\end{align*}
Thus we have $\Zm(W^-)\leq 2\Zm(W)+(q-2)\Zm(W)^2\leq q\Zm(W)$, which implies (t.3). Finally, (i\&t.2) follows from \eqref{Eq:rateReliability4} and the relation $\Zm(W)\leq qZ(W)$.
\end{proof}

\subsection{Arbitrary input alphabet sizes}
The proof of Lemma~\ref{lem:Zmax}, and hence of Theorem~\ref{thm:prime}, depends critically on the assumption that $q$ is a prime number, and does not extend trivially to the case of composite input alphabet sizes. In fact, it is possible to find channels that the transformation given in the previous section will not polarize:
\begin{example}
Consider the quaternary-input channel $W\colon\{0,1,2,3\}\to\{0,1\}$ defined by the transition probabilities $W(0\mid0)$$=W(0\mid2)=W(1\mid1)=W(1\mid3)=1$,
with $I(W) = \log 2$. If $W$ is combined/split using the transformation described in \eqref{eq:q-combine} and \eqref{eq:q-split}, where $+$ denotes modulo-$4$ addition, then the channels $W^+$ and $W^-$ are statistically equivalent to $W$. Therefore $I(W^-)=I(W)=I(W^+)$. 
\end{example}

For the general case, our first attempt at finding a polarizing transformation is to let 
\begin{align*}
x_1&=u_1+u_2\\
x_2&=\pi(u_2)
\end{align*}
where `$+$' denotes the group operation, and $\pi$ is a fixed permutation on $\cX$.  In this case one can compute easily
that
$$
Z(W^+)=\frac1{q(q-1)}\sum_{x,x':x\neq x'} Z(W_{\{\pi(x),\pi(x')\}})
	\frac1q\sum_{u} Z(W_{\{u+x,u+x'\}}).
$$
To be able to mimic the proof of Proposition~\ref{prop:bipolar} one would
want that $Z(W^+)=Z(W)^2$. However, as the value
of the inner sum above may depend on $(x,x')$, the equality $Z(W^+)=Z(W)^2$ will not necessarily hold in general.

As we will see, however, the average value of the above $Z(W^+)$ over all possible
choices of $\pi$ is $Z(W)^2$.  For this reason, it is appropriate
to think of a randomized channel combining/splitting operation, where the
randomness is over the choice of $\pi$.  To accomodate this randomness, again let $(U_1,U_2)$ denote the independent and uniformly distributed inputs, and let
$\Pi$ be chosen uniformly at random from the set of permutations
$\cP_\cX$, independently of $(U_1,U_2)$, and revealed to the receiver.
Set
\begin{equation}
\label{eq:q-transform}
(X_1,X_2)=(U_1+U_2,\Pi(U_2)).
\end{equation}
 Observe that
\begin{align*}
I(U_1,U_2;Y_1,Y_2,\Pi) &= 2I(W)\\
	&=I(U_1;Y_1,Y_2,\Pi)+I(U_2;Y_1,Y_2,U_1,\Pi),
\end{align*}
and 
that
we may define the channels $W^-\colon \cX \to \cY^2 \times \cP_\cX$ and
$W^+\colon \cX \to \cY^2 \times \cX \times \cP_\cX$ so that the terms on
the right hand side equal $I(W^-)$ and $I(W^+)$:
\begin{align}
\label{eq:qw-}
W^-(y_1,y_2, \pi \mid u_1) &= \sum_{u_2\in \cX} \frac{1}{q\cdot q!} W_2(y_1,y_2\mid u_1,u_2)\\
\label{eq:qw+}
W^+(y_1,y_2,u_1, \pi \mid u_2) &= \frac{1}{q\cdot q!} W_2(y_1,y_2\mid u_1,u_2),
\end{align}
where $W_2(y_1,y_2\mid u_1,u_2)=W(y_1\mid u_1+u_2) W(y_2\mid \pi(u_2))$.

\begin{theorem}
The transformation described in~\eqref{eq:q-transform},
\eqref{eq:qw-}, and \eqref{eq:qw+} polarizes all discrete memoryless channels $W$ in the sense of Proposition~\ref{prop:bipolar}.
\end{theorem}

\begin{proof}
As in the binary case, we will let $B_1,B_2,\dotsc$ be i.i.d., $\{+,-\}$-valued random variables taking the two values with equal probability, and define
\begin{align*}
I_n &:= I_n(B_1,\dotsc,B_n) = I(W^{B_1,\dotsc,B_n}), \\
T_n &:= T_n(B_1,\dotsc,B_n) = Z(W^{B_1,\dotsc,B_n}),
\end{align*}
with $I_0=I(W)$ and $T_0=Z(W)$. We will prove the theorem by showing that the processes $\{I_n\}$ and $\{T_n\}$ satisfy the conditions of Lemma~\ref{lem:conv}. 
Since (i.1), (i.2), (t.1) are readily seen to hold, and
(i\&t.1) is implied by inequalities~\eqref{Eq:rateReliability} and~\eqref{Eq:rateReliability2} in Proposition~\ref{rateReliability}, we only need to
show (t.2).  To that end observe that
$$
Z(W^+)=\frac1{q(q-1)}\sum_{x,x':x\neq x'}
	\frac1{q!}\sum_{\pi} Z(W_{\{\pi(x),\pi(x')\}})
	\frac1q\sum_u Z(W_{\{u+x,u+x'\}}).
$$
Note that for any $x,x'$ the value of $\frac1{q!}\sum_\pi
Z(W_{\{\pi(x),\pi(x')\}})$ is equal to $Z(W)$, and for any $u$, the value
of $\frac1{q(q-1)}\sum_{x,x'}Z(W_{\{u+x,u+x'\}})$ also equals $Z(W)$.
Thus, $Z(W^+)=Z(W)^2$.
\end{proof}

As $Z(W)$ upper bounds the error probability of uncoded transmission (cf.\ Proposition~\ref{Prop:MLerrorBound}), in order to bound the error probability of $q$-ary polar codes it suffices to show that the hypotheses of Lemma~\ref{lem:rate} hold. Since (i\&t.2) is already implied by~\eqref{Eq:rateReliability4}, it remains to show (t.3):

\begin{proposition}
\label{prop:z-}
For the transformation described in~\eqref{eq:q-transform}, \eqref{eq:qw-}, and \eqref{eq:qw+}, we have
$$
Z(W)\leq Z(W^-) \leq \min \left\{qZ(W), 2Z(W)+(q-1)Z(W)^2 \right\}.
$$
\end{proposition}

Proof is given in the Appendix.

We have seen that choosing the transformation $W\mapsto(W^-,W^+)$ in a random fashion from a set of transformations of size $q!$ yields $Z(W^+)=Z(W)^2$, leading to channel polarization. In particular, for each $W$ there is at least one transformation with $Z(W^+)\leq Z(W)^2$. Therefore, randomness is needed only in order to find such transformations at code construction stage, and not for encoding/decoding.

In a channel polarization construction of size $N$, there are $(2N-1)$ channels ($W$, $W^-$, $W^+$, $W^{--}$, $W^{-+}$, etc.) in the  recursion tree of code construction. For each channel $\tilde{W}$ residing in any one of the $(N-1)$ internal nodes of this tree, we need to find a suitable permutation $\pi$ such that $Z(\tilde{W}^+)\leq Z(\tilde{W})^2$. Thus, the total complexity of finding the right permutations scales as $q!(N-1)$, in the worst case where all $q!$ permutations are considered. Recall that polar code construction also requires determining the {\sl frozen} coordinates, which is a task of complexity $\Omega(N)$ at best. So, the order of polar code construction complexity is not altered by the introduction of randomization.

\section{Complementary Remarks}
\subsection{Reduction of randomness}

The transformation $(u_1,u_2)\mapsto(x_1,x_2)$ described above uses
a random permutation to satisfy $Z(W^+)=Z(W)^2$.  This amount of
randomness --- over a set of size $q!$ --- is in general not necessary,
randomization over a set of size $(q-1)!$ is sufficient:
\begin{theorem}
If the random permutation $\Pi$ that defines \eqref{eq:q-transform} is
chosen uniformly over the set of permutations for which $0$ is a fixed
point, the resulting transformation yields $Z(W^+)=Z(W)^2$ and thus
is polarizing.
\end{theorem}

A more significant reduction in randomness can be attained when
the input alphabet $\cX$ can be equipped with operations
$(+,\cdot)$ to form an algebraic field --- this is possible if and only if $q$ is a prime power. A random variable taking only
$q-1$ values is sufficient in this case. (We have already seen that no randomization is needed when $q$ is prime.) To see this, pick $R$ to be uniformly distributed
from the non-zero elements $\cX_*$ of $\cX$, reveal it to the receiver and set
\begin{equation}
\label{eq:f-transform}
(x_1,x_2)=(u_1+u_2, R\cdot u_2).
\end{equation}
As was above we have
$$
2I(W)=I(U_1,U_2;Y_1,Y_2,R)
=I(U_1;Y_1,Y_2,R)+I(U_2;Y_1,Y_2,U_1,R)= I(W^-)+I(W^+)
$$
provided that we define $W^-\colon\cX\to\cY^2\times\cX_*$ and $W^+:\cX\to\cY^2\times\cX\times\cX_*$ as
\begin{align}
\label{eq:fw-}
W^-(y_1,y_2,r|u_1)&=\frac1{q(q-1)}\sum_{u_2\in\cX} W(y_1|u_1+u_2)W(y_2|r\cdot u_2),\\
\label{eq:fw+}
W^+(y_1,y_2,u_1,r|u_2)&=\frac1{q(q-1)} W(y_1|u_1+u_2)W(y_2|r\cdot u_2).
\end{align}
\begin{theorem}
The transformation described in~\eqref{eq:f-transform},
\eqref{eq:fw-}, and \eqref{eq:fw+} polarizes all $q$-ary input channels in the sense of Proposition~\ref{prop:bipolar}, provided that $q$ is a prime power.
\end{theorem}
\begin{proof}
Again, we only need to show that $Z(W^+)=Z(W)^2$.  To that end observe that
$$
Z(W^+)=\frac1{q(q-1)}\sum_{x,x'\colon x\neq x'}
	\frac1{q-1}\sum_{r\neq 0} Z(W_{\{r\cdot x,r\cdot x'\}})
	\frac1q \sum_u Z(W_{\{u+x,u+x'\}}).
$$
Writing $x'=x+d$, and $u'=u+x$, we can rewrite the above as
$$
Z(W^+)=\frac1{q^2(q-1)^2}
\sum_{d\neq 0}
\sum_x
\sum_{r\neq 0}
	Z(W_{\{r\cdot x,r\cdot x+r\cdot d\}})\sum_{u'}Z(W_{\{u',u'+d\}})
$$
Noting that for any fixed $d$, the sum
$\frac1{q(q-1)}\sum_{x,r\neq0}Z(W_{\{r\cdot x,r\cdot x+r\cdot d\}})$ equals $Z(W)$, and that the sum $\frac1{q(q-1)}\sum_{u',d\neq0}Z(W_{\{u',u'+d\}})$ also
equals $Z(W)$ yields $Z(W^+)=Z(W)^2$.
\end{proof}

When the field is of odd characteristic (i.e., when $q$ is not a power of two),
a further reduction is possible: since $\sum_{u'}Z(W_{\{u',u'+d\}})$ is invariant
under $d\to -d$, one can show that the range of $R$ can be reduced from
$\cX_*$ to only half of the elements in $\cX_*$, by partition $\cX_*$ into
two equal parts in one-to-one correspondence via $r\mapsto-r$, and
picking one of the parts as the range of $R$. It is easy to show that
choosing $R$ uniformly at random over this set of size $(q-1)/2$ will
also yield $Z(W^+)=Z(W)^2$.

\subsection{A method to avoid randomness}\label{subsec:AvoidRandomness}
When the input alphabet size $q$ is not prime, an alternative multi-level code construction technique can be used in order to avoid randomness:
Consider a channel $W$ with input alphabet size $q=\prod_{i=1}^L q_i$, where $q_i$'s are the prime factors of $q$. When the input $X$ to $W$ is uniformly distributed on $\cX$, one can write $X = (U_1,\dotsc,U_L)$, where $U_i$'s are independent and uniformly distributed on their respective ranges $\cU_i = \{0,\dotsc,q_i-1\}$. Defining the channels $W^{(i)}: \cU_i \to \cY\times\cU_1\times\dotsc\times\cU_{i-1}$ through
$$
W^{(i)}(y,u_1^{i-1}\mid u_i) = \prod_{j\neq i}q_j^{-1} \sum_{u_{i+1}^L} W(y|(u_1^L)),
$$
it is easily seen that
$$
I(W) = I(X;Y) =  I(U_1^L;Y) = \sum_i I(U_i;Y,U_1^{i-1}) = \sum_i I(W^{(i)}).
$$

Having decomposed $W$ into $W^{(1)},\dotsc,W^{(L)}$, one can polarize each channel $W^{(i)}$ separately. The order of successive cancellation decoding in this multi-level construction is to first decode all channels derived from $W^{(1)}$, then all channels derived from $W^{(2)}$, and so on. Since the input alphabet size of each channel is prime, no randomization is needed.

\subsection{Equidistant channels}
A channel $W$ is said to be {\sl equidistant} if $Z(W_{\{x,x'\}})$ is constant for all pair of distinct input letters $x$ and $x'$.
These are channels with a high degree of symmetry. In particular, if a channel $W$ is equidistant, then so are the channels $W^+$ and $W^-$ created by the deterministic mapping $(u_1,u_2)\mapsto(u_1+u_2,u_2)$. By similar arguments to those in Section~\ref{subsec:prime}, it follows that this mapping polarizes equidistant channels, regardless of the input alphabet size.

\subsection{How to achieve channel capacity using polar codes}\label{subsec:NonUniform}
In all of the above, the input letters of the channel under consideration were used with equal frequency. This was sufficient to achieve the symmetric channel capacity. However, in order to achieve the true channel capacity, one should be able to use the channel inputs with non-uniform frequencies in general. The following method, discussed in \cite[p.\ 208]{Gallager}, shows how to implement non-uniform input distributions within the polar coding framework.

Given two finite sets $\cX$ and $\cX'$ with $m=|\cX'|$,
any distribution $P_X$ on $\cX$ for which $mP_X(x)$ is an integer for
all $x$ can be induced by the uniform distribution on $\cX'$
and a deterministic map $f:\cX'\to\cX$.

Given a channel $W:\cX\to\cY$, and a distribution $P_X$ as above,
we can construct the channel $W':\cX'\to \cY$ whose input alphabet is
$\cX'$ and $W'(y|x')=W(y|f(x'))$.  Then $I(W')$ is the same as the
mutual information developed between the input and output of the channel
$W$ when the input distribution is $P_X$.  Consequently, a method
that achieves the symmetric capacity of any discrete memoryless channel, such
as the channel polarization method considered in this paper,
can be extended to approach the true capacity of any discrete memoryless
channel by taking $P_X$ as a rational distribution approximating the
capacity achieving distribution. (In order to avoid randomization, one may use prime $m$ in the constructions.)

\subsection{Channels with continuous alphabets}

Although the discussion above has been restricted to channels with {\sl discrete} input and output alphabets, it should be clear that the results hold when the output alphabet is continuous, with minor notational changes.
In the more interesting case of channels with continuous input alphabets --- possibly with input constraints, such as the additive Gaussian noise channel with an input power constraint --- we may readily apply the method of Section~\ref{subsec:NonUniform} to approximate any desired continuous input distribution for the target channel, and thereby approach its capacity using polar codes.

\section*{Acknowledgment}
The work of E. Ar{\i}kan was supported in part by The Scientific and
Technological Research Council of Turkey (T\"UB\.ITAK) under
contract no. 107E216, and in part by the European Commission
FP7 Network of Excellence NEWCOM++ (contract no. 216715).

\appendix \label{Appendix}
\subsection{Proof of Proposition~\ref{rateReliability}}

This proposition was proved in \cite{Arikan2009} for the binary case $q=2$. Here, we will reduce the general case to the binary case.

\subsubsection{Proof of \eqref{Eq:rateReliability}}
The right hand side (r.h.s.) of \eqref{Eq:rateReliability} equals the channel parameter known as {\sl symmetric cutoff rate\/}.
More specifically, it equals the function $E_0(1,Q)$ defined in Gallager \cite[Section~5.6]{Gallager} with $Q$ taken as the uniform input distribution. It is well known (and shown in the same section of \cite{Gallager}) that the cutoff rate cannot be greater than $I(W)$. This completes the proof of \eqref{Eq:rateReliability}.

\subsubsection{Proof of \eqref{Eq:rateReliability4}}

\begin{lemma}\label{rateReliability3}
For any $q$-ary channel $W \colon \cX\to \cY$,
\begin{align}\label{Eq:rateReliability3}
I(W) \le \log(q/2) + \sum_{x_1,x_2\in \cX: x_1\neq x_2} \frac{1}{q(q-1)}I(W_{\{x_1,x_2\}}).
\end{align}
\end{lemma}

\begin{proof}
Let $(X,Y,X_1,X_2)\sim Q(x)P(x_1,x_2|x)W(y|x)$ where $Q$ is the uniform distribution on $\cX$ and
\begin{align*}
P(x_1,x_2|x) = \begin{cases}
\frac{1}{2(q-1)} & \text{if $x_1=x$ and $x_2\neq x_1$}\\
\frac{1}{2(q-1)} & \text{if $x_2=x$ and $x_1\neq x_2$}\\
0 & \text{otherwise}
\end{cases}
\end{align*}
Clearly we have $I(W) = I(X;Y)\le I(X;Y,X_1,X_2)$.
By the chain rule,
$I(X;Y,X_1,X_2) = I(X;X_1,X_2) + I(X;Y|X_1,X_2)$.
Now, simple calculations show that $I(X;X_1,X_2)$ and $I(X;Y|X_1,X_2)$ equal the two terms that appear
on the right side of \eqref{Eq:rateReliability3}.
(Intuitively, $(X,Y)$ are the input and output of $W$ and $(X_1,X_2)$ is a side information of value $\log(q/2)$ supplied by a genie to the receiver.)
\end{proof}

Note that the summation in \eqref{Eq:rateReliability3} can be written as the expectation
$E\left[I(W_{\{X_1,X_2\}})\right]$ where $(X_1,X_2)$ ranges over all distinct pairs of letters from $\cX$ with equal probability.
Next, use the form of \eqref{Eq:rateReliability4} for $q=2$ (which is already established in \cite{Arikan2009}) to write $E\left[I(W_{\{X_1,X_2\}})\right] \le \log(2)\,
E\left[\sqrt{1-Z(W_{\{X_1,X_2\}})^2}\right]$.
Use Jensen's inequality on the function $\sqrt{1-x^2}$, which is concave for $0\le x\le 1$, to obtain
$E\left[\sqrt{1-Z(W_{\{X_1,X_2\}})^2}\right] \le \sqrt{1-E[Z(W_{\{X_1,X_2\}})]^2}$. Since $Z(W) = E[Z(W_{\{X_1,X_2\}})]$,
this completes the proof of \eqref{Eq:rateReliability4}.

\subsubsection{Proof of \eqref{Eq:rateReliability2}}

For notational simplicity we will let $W_x(\cdot) := W(\cdot\mid x)$. First note that 
$$
I(W) = \frac 1q \sum_{x\in\cX} D\left( W_x \bigg\| \frac1q \sum_{x'} W_{x'} \right)
$$
where $D(\cdot\| \cdot)$ is the Kullback-Leibler divergence.
Each term in the above summation can be bounded as
\begin{align}
\notag
D\left( W_x \bigg\| \frac1q \sum_{x'} W_{x'} \right) &
	= \sum_y W_x(y) \log \frac{W_x(y)}{\frac1q \sum_{x'}W_{x'}(y)} \\
\notag
& \leq \log e\sum_y W_x(y) \left( \frac{W_x(y)-\frac1q \sum_{x'}W_{x'}(y)}{\frac1q \sum_{x'}W_{x'}(y)} \right) \\
\notag
& \leq q\log e\sum_y \left|W_x(y) - \frac1q \sum_{x'}W_{x'}(y) \right| \\
\label{eq:inv-pinsker}
& = q\log e\left\| W_x - \frac1q \sum_{x'} W_{x'}\right\|_1.
\end{align}
In the above, the first inequality follows from the relation $\ln(x)\leq x-1$, and the second inequality is due to $W_x(y) \leq \sum_{x'} W_{x'}(y)$. The $\mathcal{L}_1$ distance on the right hand side of \eqref{eq:inv-pinsker} can be bounded, using the triangle inequality, as
$$
\left\| W_x -\frac1q \sum_{x'}W_{x'}\right\|_1 \leq 
	\frac1q \sum_{x'\in\cX} \left\| W_x - W_{x'} \right\|_1.
$$
Also, it was shown in \cite[Lemma 3]{Arikan2009} that
$$
\left\| W_x - W_{x'} \right\|_1 \leq 2\sqrt{1-Z(W_{\{x,x'\}})^2}.
$$
Combining the inequalities above, we obtain
\begin{align*}
I(W) 
& \leq \frac{2\log e}{q} \sum_{x,x'\in\cX\colon x\neq x'} \sqrt{1-Z(W_{\{x,x'\}})^2} \\
& \leq 2(q-1) \log e \sqrt{1-Z(W)^2},
\end{align*}
where the last step follows from the concavity of the function $x\mapsto\sqrt{1-x^2}$ for $0\leq x\leq1$.

\subsection{Proof of Proposition~\ref{prop:z-}}
Define the channel $W^{(\pi u)}$ through
$$
W^{(\pi u)}(y_1y_2\mid x) = W(y_1\mid x+u)W(y_2\mid \pi(u)).
$$
and let
$$
W^{(\pi)} = \frac1q \sum_{u\in\cX} W^{(\pi u)}.
$$
Note if one fixes the permutation in the transformation $W\mapsto (W^-,W^+)$ to $\pi$, then $W^-=W^{(\pi)}$.

We will show the stronger result that
\begin{align*}
Z(W)\leq Z(W^{(\pi)}) & \leq \min \{qZ(W), 2Z(W)+(q-1)Z(W)^2\}
\end{align*}
for all $\pi$, which will imply Proposition~\ref{prop:z-} since $Z(W^-) = \frac{1}{q!} \sum_\pi Z(W^{(\pi)})$. To prove the upper bound on $Z(W^{(\pi)})$, we write
\begin{align}
Z(W^{(\pi)}) & = \frac{1}{q(q-1)} \sum_{\substack{x,x' \in \cX \\ x\neq x'}} \sum_{y_1,y_2\in \cY} \frac{1}{q} \sqrt{\sum_{u\in \cX} W(y_2\mid \pi(u))W(y_1\mid x+u)} \sqrt{\sum_{v\in \cX} W(y_2\mid \pi(v))W(y_1\mid x'+v)} \notag \\
& \leq \frac{1}{q(q-1)} \sum_{\substack{x,x' \\ x\neq x'}} \sum_{y_1,y_2} \frac{1}{q} \sum_{u} \sqrt{W(y_2\mid \pi(u))W(y_1\mid x+u)} \sum_{v} \sqrt{W(y_2\mid \pi(v))W(y_1\mid x'+v)} \notag \\
& =  \frac{1}{q} \sum_{u} \frac{1}{q(q-1)} \sum_{\substack{x,x' \\ x\neq x'}} \sum_{y_2} W(y_2\mid \pi(u)) \sum_{y_1} \sqrt{W(y_1\mid x+u)W(y_1\mid x'+u)} \label{eqn:Z2-1}\\
& \qquad + \frac{1}{q^2(q-1)} \sum_{\substack{u,v\\ u\neq v}} \sum_{y_2} \sqrt{W(y_2\mid \pi(u))W(y_2\mid \pi(v))} \sum_{\substack{x,x' \\ x\neq x'}} \sum_{y_1} \sqrt{W(y_1\mid x+u)W(y_1\mid x'+v)}. \label{eqn:Z2-2}
\end{align}
Note that
$$
\sum_{y_2} W(y_2\mid \pi(u)) \sum_{y_1} \sqrt{W(y_1\mid x+u)W(y_1\mid x'+u)} = Z(W_{\{x+u,x'+u\}})
$$
for any $u\in \cX$. Therefore the r.h.s.\ of (\ref{eqn:Z2-1}) is equal to $Z(W)$. Also, note that the innermost sum over $y_1$ in \eqref{eqn:Z2-2} is upper bounded by $1$. Therefore, \eqref{eqn:Z2-2} is upper bounded by $(q-1)Z(W)$. Alternatively, noting that for any fixed $u\neq v$
\begin{align*}
\sum_{\substack{x,x'\\ x\neq x'}} \sum_{y_1} \sqrt{W(y_1\mid x+u)W(y_1\mid x'+v)} & = q + \left[\sum_{\substack{x,x'\colon x\neq x' \\ x+u \neq x'+v}} \sum_{y_1} \sqrt{W(y_1\mid x+u)W(y_1\mid x'+v)} \right] \\
& \leq q + q(q-1) Z(W),
\end{align*}
we have 
\begin{align*}
\text{r.h.s.\ of (\ref{eqn:Z2-2}) } & \leq (1+(q-1)Z(W)) \frac{1}{q(q-1)} \sum_{\substack{u,v\\ u\neq v}} \sum_{y_2} \sqrt{W(y_2\mid u)W(y_2\mid v)} \\
& = Z(W)+(q-1)Z(W)^2.
\end{align*}
This in turn implies $Z(W^{(\pi)}) \leq \min \left\{qZ(W), 2Z(W) + (q-1) Z(W)^2 \right\}$. 

The proof of $Z(W)\leq Z(W^{(\pi)})$ follows from the concavity of $Z(W_{\{x,x'\}})$ in $W$, shown in \cite{Arikan2009}:
\begin{align*}
Z(W^{(\pi)}) & = \frac{1}{q(q-1)} \sum_{x\neq x'} Z(W_{\{x,x'\}}^{(\pi)}) \\
& \geq \frac{1}{q(q-1)} \sum_{x\neq x'} \frac1q \sum_u Z(W^{(\pi u)}_{\{x,x'\}}) \\
& = \frac1q \sum_u \frac{1}{q(q-1)} \sum_{x\neq x'}  \sum_{y_1,y_2} \sqrt{W(y_1\mid x+u))W(y_1\mid x'+u)W(y_2\mid \pi(u))W(y_2\mid \pi(u))} \\
& = \frac1q \sum_u \frac{1}{q(q-1)} \sum_{x\neq x'} Z(W_{\{x+u,x'+u\}}) \\
& = Z(W).
\end{align*}

\end{document}